\documentclass{article}[12pt]
\usepackage{fullpage}
\usepackage{amsmath}
\usepackage{amsfonts}

   \newtheorem{theorem}{Theorem}[section]
   \newtheorem{lemma}[theorem]{Lemma}
   
   \newtheorem{corollary}[theorem]{Corollary}

\newcommand{\be}{\begin{equation}} 
\newcommand{\ee}{\end{equation}}
\newcommand{\beq}{\begin{eqnarray}}
\newcommand{\eeq}{\end{eqnarray}}
\newcommand{\bt}{\begin{theorem}}
\newcommand{\et}{\end{theorem}}
\newcommand{\bl}{\begin{lemma}}
\newcommand{\el}{\end{lemma}}
\newcommand{\bc}{\begin{corollary}}
\newcommand{\ec}{\end{corollary}}
\newcommand{\bp}{\begin{prop}}
\newcommand{\ep}{\end{prop}}
\newcommand{\ba}{\begin{array}}
\newcommand{\ea}{\end{array}}

\newcommand{\comm}[1]{}

\newcommand{\la}{\label}
\newcommand{\ci}{\cite}
\newcommand \qed {\hskip 6pt\vrule height6pt width5pt depth1pt \bigskip}

\newcommand{\lb}{\lambda}

\newfont{\msbm}{msbm10 scaled\magstep1}
\newfont{\msbms}{msbm7 scaled\magstep1} 

\newcommand{\bbr}{\mbox{$\mbox{\msbm R}$}}

   \newenvironment{proof}[1][Proof]{\begin{trivlist}
   \item[\hskip \labelsep {\bfseries #1}]}{\end{trivlist}}

   \newenvironment{remark}[1][Remark]{\begin{trivlist}
   \item[\hskip \labelsep {\bfseries #1}]}{\end{trivlist}}

   \comm{\newcommand{\qed}{\nobreak \ifvmode \relax \else
      \ifdim\lastskip<1.5em \hskip-\lastskip
      \hskip1.5em plus0em minus0.5em \fi \nobreak
      \vrule height0.75em width0.5em depth0.25em\fi}}

   \numberwithin{equation}{section}

\newcommand{\zee}{\mathbb{Z}}       
\newcommand{\real}{\mathbb{R}} 
\newcommand{\nn}{\mathbb{N}}  %

\newcommand{\cx}{\mathbb{C}} 

\newcommand{\de}{\mathrm{d}} 

\newcommand{\tor}{\mathbb{T}}

\newcommand{\slr}{  \textrm{ {\bf SL}}_2 (\mathbb{R})  }

\newcommand{\ly}{\mathcal{L}}

\newcommand{\ham}{H}

\begin{document}
\title{Continuity of the measure of the spectrum for quasiperiodic
            Schr\"{o}dinger operators with rough potentials}
\author{Svetlana Jitomirskaya
 and Rajinder Mavi\thanks{The work was supported by NSF Grant DMS-1101578.}
}

\maketitle
%

\begin{abstract}
We study discrete quasiperiodic  Schr\"odinger operators on $\ell^2(\zee)$ with potentials defined by $\gamma$-H\"older functions. We prove a general statement  that  for $\gamma >1/2$  and under the condition of positive Lyapunov exponents, measure of the spectrum at irrational frequencies is the limit of measures of spectra of periodic approximants. 
An important ingredient in our analysis is  a general result on uniformity of the upper Lyapunov exponent of strictly ergodic cocycles.
\end{abstract}

\section{Introduction} Consider quasiperiodic  operators acting on $l^2(\zee)$ and given by:
\begin{equation}\label{schrodinger1}
(\ham_{\omega,\theta}\psi)(n)=
\psi(n-1)+\psi(n+1)+f(\omega n+\theta)\psi(n),
\qquad n=\dots,-1,0,1,\dots,
\end{equation}
where $f(x)$ is a real-valued sampling function of period 1. 
Denote by $S(\omega,\theta)$ the spectrum of $\ham_{\omega,\theta}.$ 
For rational $\alpha=p/q$ the spectrum 
consists of at most $q$ intervals. Let $S(\omega)= \bigcup_{\theta\in\bbr}
S\left(\omega,\theta\right).$  Note that for irrational $\omega$
the spectrum of $H$ (as a set) is independent of $\theta$ (see, e.g., 
\ci{cfks}), and therefore $S(\omega,\theta)=S(\omega).$ In this paper 
we study continuity of $S(\omega)$ and its
measure upon rational approximation of $\omega,$  for rough sampling functions $f.$ 

The last decade has seen an explosion of general results for operators
(\ref{schrodinger1}) with analytic $f,$ see
e.g. \cite{BourgainGreen,J} and references therein, and by now even
the global theory of such operators is well developed
\cite{Aglobal1,Aglobal2}.  There are very few complete results,
however,  beyond the analytic category. Indeed, not only the methods
of the mentioned papers intrinsically require analyticity or at least
Gevrey regularity (e.g. the large deviation theorems), but it is essential for some results too. For example, continuity of the Lyapunov exponent \ci{BJ02},  an important ingredient of many later developments, 
 may not hold in the case of even $C^\infty$ regularity \cite{WY2012}
 (see also \ci{JM1}). The surprising counterexample in \cite{WY2012}
 has made it
 natural to conjecture that (near) analyticity is essential for many
 other general properties of  quasiperiodic potentials: positive
 Lyapunov exponents at high couplings, localization in the regime of
 positive Lyapunov exponents, finiteness of transitions between
 supercritical and subcritical regimes, almost reducibility
 conjecture, etc. This paper presents a result in the opposite direction. 
We show that, under certain conditions, for the fundamental question
of continuity in $\omega$ (previously established under the
analyticity condition) not only analyticity is not essential, but
such continuity always holds even at surprisingly weak
regularity. Namely, in the regime of positive Lyapunov exponents,
spectra of rational approximants converge a.e. to $S(\omega)$  for all
$f,$ with H\"older-1/2+ continuity. To our knowledge, other than the
very basic facts that require, at most, continuity of $f,$  there are
no other results that do not require exclusion of potentially relevant
parameteres or additional assumptions (e.g. transversality) and work
for potentials that rough, and the fact that one can even go beyond
the Lipshitz condition has been a surprise to the authors. Moreover, we have reasons to expect that our condition is
optimal as far as H\"older regularity go.

The fact that various quantities could be easier to analyze and sometimes 
are even computable
for periodic operators, $H_{p/q,\theta}$,  makes results on continuity in 
$\omega$ particularly important.  For example, the famous Hofstadter butterfly \cite{H1976} is a plot of the almost Mathieu spectra for 50 rational values of $\omega$ and visually based inferences about the spectrum for irrational $\omega$ implicitly assume continuity. It is therefore an important and natural question if and in what sense the spectral properties of such rational approximants relate to those of the quasi-periodic operator $H_{\omega,\theta}.$

The history of this question was centered around the Aubry conjecture on the measure of the 
spectrum \ci{AA1980}, popularized by B. Simon \cite{Simon84,Simon00} : that for the almost Mathieu operator given by
(\ref{schrodinger1}) with $f(\theta)=2\lb\cos2\pi\theta$, for irrational $\omega$ and 
all real $\lb,\theta$ there is an equality
\be
\la{aa}
|S_\lambda(\omega,\theta)|=4|1-|\lb||.
\ee
Here, for sets, we use $|\cdot |$ to denote the Lebesgue measure. 
Avron,van Mouche, Simon \ci{AMS1990} proved that, for $|\lb| \not= 1,$ 
$|S_\lambda(p_n/q_n)|\to 4|1-|\lb||$ as $q_n\to \infty$, and
Last \ci{L1994} established this fact for $|\lb| = 1$. Given these theorems, 
the proof of the Aubry-Andre conjecture was reduced to a
continuity result.
   
   The continuity in $\omega$ of $S(\omega)$ in Hausdorff metric was
   proven in \ci{AS1983,El82}, requiring only the continuity of $f.$
Continuity of the measure of the spectrum is  a more delicate issue, since, 
in particular, $|S(\omega)|$ can be (and is, for the almost Mathieu operator)  discontinuous at
rational $\omega.$ We will actually use an even stronger notion of a.e. setwise continuity. Namely, we say  $\lim_{n\to \infty} B_n  = B$ if and only if 

\begin{equation} \label{eq_def_topology}
\limsup_{n \to \infty} B_n = \liminf_{n \to \infty} B_n = B 
\Longleftrightarrow  \lim_{n \to \infty} \chi_{B_n} 
 = \chi_B ~\mbox{Lebesgue a.e. }
\end{equation}
.

Establishing continuity at irrational $\omega$ requires 
quantitative estimates on the Hausdorff continuity of the
spectrum. The first such result, namely the H\"older-$\frac{1}{3}$ 
continuity was proved in
\ci{CEY}, where it was used to establish a zero-measure spectrum  
(and therefore the Aubry-Andre conjecture) for the almost Mathieu
operator with Liouville frequencies $\omega$ at the critical 
coupling $\lb =1$. That argument was improved to the H\"older-1/2 continuity
(for arbitrary
$f \in C^1$) in
\ci{AMS1990} and subsequently used in \ci{L1993,L1994} to establish (\ref{aa})
 for the almost Mathieu operator 
for $\omega$ with unbounded continuous fraction expansion, therefore proving 
the Aubry-Andre conjecture for a.e. (but not all)
$\omega.$  The extension to all irrational $\omega$ is due to \cite{JK2001,AK2006}\footnote{ It should be noted that the argument of \cite{AK2006} that, in particular, completed the result for the critical value of $\lambda$, did not involve continuity in frequency}.

It was argued in \ci{AMS1990} that H\"older continuity of any order larger than $1/2$ would imply 
the desired continuity property of the measure of the
spectrum for all $\omega$ and $f \in C^1.$  It was first noted in \ci{JL1998} that in the regime of semi-uniform localization, the appropriate cut-offs of the exponentially localized eigenfunctions provide good enough approximate eigenvectors for a perturbed operator to establish almost Lipshitz continuity (thus establishing the Aubry-Andre conjecture in the localization regime available at that time). The idea of \ci{JK2001} was  that for Diophantine $\omega$ and analytic $f$  one can extract such eigenvectors (and thus establish almost Lipshitz continuity of $S$) by finding the cut-off places at distance $L$ from each other where the generalized eigenfunction is exponentially small in $L,$  simply as a corollary of positive Lyapunov exponents, without establishing localization. This led to  establishing that, in the regime of positive Lyapunov exponents, for any analytic $f,$  $|S(\frac{p_n}{q_n})| \to |S(\omega)|$ for every Diophantine $\omega$ and its approximants $\frac{p_n}{q_n}.$  Recently, it was shown in \ci{JM2012} that positivity of the Lyapunov exponent is not needed for this result, in particular, for analytic $f,$ and all irrational $\omega,$ $S(\frac{p_n}{q_n})\to S(\omega).$

Our goal is to show that under the condition of positivity of the
Lyapunov exponent, one can significantly relax the required regularity of $f.$

   For a given energy $E\in\real$, a formal solution $u$ of 
\be \label{schrodinger2} \ham u = Eu\ee
 with operator $\ham$ given by (\ref{schrodinger1}) 
  \comm{ \begin{equation}\label{schrodinger2}
      u(n+1) + u(n-1) + f(\theta + n\omega)u(n) = E u(n)
   \end{equation}}
  can be reconstructed from its values at two consecutive points with the transfer matrix
   \begin{equation}\label{evop}
     A^E(\theta) = 
     \left(\begin{array}{cc}E - f(\theta) & -1\\ 1&0 \end{array} \right); \hspace{.5in} A^E: \tor \to \slr
   \end{equation}
   via the equation
   \begin{equation}\label{dynsol}
      \left(\begin{array}{c}u(n+1)\\ u(n)\end{array}\right) =
     A^E(\theta+n\omega)      \left(\begin{array}{c}u(n)\\ u(n-1) \end{array} \right).
   \end{equation}
   Setting  $R:\tor\to\tor,\;Rx:=x+\omega,$ the pair $(\omega, A^E)$ viewed as a linear skew-product 
   $(x,v)\to(Rx,A(x)v),\;x\in\tor,\;v\in\bbr^2,$ is called the corresponding Schr\"{o}dinger cocycle.
   The  iterations of the cocycle $(\omega,A^E)$ for
   $k\geq 0$ are given by
   \begin{equation}\label{cocycle}
          A^E_{k}(\theta) = A^E(R^{(k-1)}\theta) \cdots A^E(R^{ 1}\theta) A^E(\theta),
      \hspace{.2in} A_0^E = I
   \end{equation} 
   and
      \begin{equation}\label{backcocycle}
      A^E_{k}(\theta) = \left(A^E_{-k}(R^{k+1}\theta)\right)^{-1};\hspace{.1in}k<0.
   \end{equation} 
   Therefore, it can be seen from (\ref{dynsol})
   that a solution to (\ref{schrodinger2}) for chosen initial conditions
   $(u(0),u(-1))$ for all $k\in \zee$ is given by,
   \begin{equation}
      \left(\begin{array}{c}u(k)\\ u(k-1)\end{array}\right) =  
     \ A_k^E(\theta)\left(\begin{array}{c}u(0)\\ u(-1)\end{array}\right).
   \end{equation}
   From general properties of subadditive ergodic cocycles, we can define the Lyapunov
   exponent
   \begin{equation}\label{avginf}
      \ly(E) = \lim_k\frac{1}{k}\int\ln \|A^E_k(\theta)\|\de\theta
      = \inf_k \frac{1}{k}\int\ln \|A_k^E(\theta)\|\de\theta,
   \end{equation}
   furthermore,
   $\ly(E) =\lim_k\frac{1}{k}\ln \|A^E_k(\theta)\|$
   for almost all $\theta\in\tor$.

   As mentioned above, $\ly$ may be discontinuous in the non-analytic category. 
   Set $L_+(\omega):=\{E: \ly(E)>0\}.$

   Our main result is

   \begin{theorem}\label{maint}
          For every irrational $\omega$, there exists a 
	  sequence of rationals $\frac{p_n}{q_n}\to \omega$ such that
          for any  $f\in\mathcal{C}^\gamma(\tor)$
	  with $\gamma>\frac{1}{2}$ 
	  \begin{equation}\label{irrmc}
	     S\left(\frac{p_n}{q_n}\right)\cap L_+(\omega)
              \to S\left(\omega\right) \cap L_+(\omega). 
	  \end{equation}
   \end{theorem}
   \begin{remark}
    \begin{enumerate}
     \item The convergence holds in the strong sense of (\ref{eq_def_topology}).
      \item The sequence $\frac{p_n}{q_n}$ will be the full sequence 
            of continued fraction approximants of $\omega$ in the 
            Diophantine case, and an appropriate subsequence of it 
            otherwise. 
            For practical purposes of making conclusions about 
            $ S(\omega)$ based on the information on 
            $ S\left(\frac{p_n}{q_n}\right)$ it is sufficient 
            to have convergence along a subsequence.
      \item It is an interesting question whether $\gamma =1/2$ 
            represents a sharp regularity threshold for this result 
            for a.e. $\omega$.
      \item Lower regularity is sufficient for a measure zero set of 
            non-Diophantine $\omega$, see Theorem~\ref{main2}
      \item It is also interesting to find out what is the lowest 
            regularity requirement for the convergence of full union
            spectra, without condition of positive $\ly,$ and for the related 
            Last's intersection spectrum conjecture. Both are more 
            delicate and currently established only for analytic 
            $f$ (\ci{JM2012}, see also \cite{shamis}). 
            We expect that higher than $1/2$ regularity 
            should be required for those results, but that analytic 
            condition is improvable. 
      \item If we replace $f$ in (\ref{schrodinger1}) with $\lambda f$, 
            then $\ly$ is expected to be positive for most $f$ and large 
            $\lambda$ through most of the spectrum, creating a wide range 
            of applicability for Theorem~\ref{maint}. 
            For analytic $f$ this is known to hold uniformly in $(E,\omega)$ 
            for large $\lambda$. 
            For the rough case, the relevant results 
            are \cite{Spencer,bjerklov}(reviewer requests Sinai)
    \end{enumerate}
   \end{remark}
   Theorem~\ref{maint} certainly has an immediate corollary:
   \begin{corollary}\label{cor1}
      For every irrational $\omega$, there exists a 
      sequence of rationals $\frac{p_n}{q_n}\to \omega$ such that
      for any  $f\in\mathcal{C}^\gamma(\tor)$
      with $\gamma>\frac{1}{2}$ 
	  \begin{equation}\label{irrmc2}
	     |S\left(\frac{p_n}{q_n}\right)\cap L_+(\omega)|
                 \to |S\left(\omega\right) \cap L_+(\omega)|. 
	  \end{equation}
   \end{corollary}

This corollary for analytic $f$ was the main result of \cite{JK2001} and our proof borrows some important ingredients from that work. The main idea of the current paper is to show that, for Diophantine frequencies,  $\gamma$-H\"older continuity of $f$ is sufficient to find  the cut-off places at distance $L$ from each other where the generalized eigenfunction is {\em polynomially} small in $L,$ thus establishing $\beta$-H\"older continuity of the spectrum with $\beta<\gamma.$ The requirement $\gamma>1/2$ comes from the application of the original argument in \ci{AMS1990}.  For non-Diophantine frequencies we obtain the statement by extending the H\"older continuity theorem of \ci{AMS1990} in the following way: for $\gamma$-H\"older functions $f$ the spectrum is $\frac{\gamma}{1+\gamma}$-H\"older continuous, which is sufficient, under an appropriate anti-Diophantine condition, even without positivity of the Lyapunov exponent, see Theorem~\ref{arc}.

The proof requires very tight control on the perturbations of
cocycles, in absence of continuity of the Lyapunov exponent. To this
end, we show that generally, for cocycles over uniquely ergodic
dynamics, upper bound is uniform in phases and neighborhoods
(Theorem~\ref{uniformsemicontinuity}).

    The main part of the proof of Theorem~\ref{maint} follows from 
    H\"{o}lder continuity properties of the spectrum in the Hausdorff metric 
    which are stated in section \ref{ear}. 
    The argument for the positive Lyapunov exponent regime uses tight bounds on  matrix cocycle approximation covered in
    Section \ref{rocmc} which in turn depend on a general result on uniform 
    upper-semicontinuity of Lyapunov exponents for cocycles over uniquely ergodic 
    dynamics,
    and is proven in Section \ref{usule}. 
    Section \ref{com} completes the proof of Theorem~\ref{maint}.

   \comm{   Let $f$ be a continuous real valued sampling function 
   defined on the torus $\tor = \real/\zee$,
   we get a family of Schr\"{o}dinger operators on $\ell^2(\zee)$ 
   \begin{equation}\label{schrodinger1}
      (\ham_{\omega,\theta} u)(n) = u(n+1) + u(n-1) + f(\theta + n\omega)u(n).
   \end{equation}
   parametrized by $\omega,\theta\in\tor$. 
   Let $S(\omega,\theta) = \sigma(\ham_{\omega,\theta})$.
   We will consider the upper spectral envelope in this paper,
   for any $\omega\in \tor$
   $S(\omega) = \cup_{\theta\in\tor} S(\omega,\theta)$; 
   for irrational $\omega$, for any $\theta\in\tor$,
   $S(\omega) = S(\omega,\theta)$.

     We will consider more general ergodic cocycles in Sections \ref{usule} and
   \ref{rocmc}.}

\section{Continued fraction approximants}
   For $\kappa \geq 0$, $\omega\in\real$ is said to be 
   $\kappa$-Diophantine if there exists some $C_\omega>0$ so that
   \begin{equation}\label{dioph}
      \|  n\omega \| > \frac{C_\omega}{n^{1+\kappa}} 
   \end{equation}
   for all $n\in\zee,$ where $\|\cdot\|$ denotes distance to the integers.
   For $\kappa >0$ a.e. $\omega\in\tor$ is $\kappa$-Diophantine.
   For $\kappa =0$ this condition is equivalent to having bounded type, 
   so a.e. $\omega\in\tor$ is not $0$-Diophantine.

   Writing $\omega$ in continued fraction expansion,
   \[ \omega = a_0 + \cfrac{1}{a_1+\cfrac{1}{a_2 +  \cfrac{1}{a_3+\ldots} } } =
        [a_0;a_1,a_2,\ldots],\]
   the truncations $p_n/q_n = [a_0;a_1,a_2,\ldots,a_n]$ 
   are known as the continued fraction approximants.  
   From the theory of continued fractions \cite{khinchin},
   for $\kappa$-Diophantine $\omega$ and $n > n_\omega$ we have 
   for some $C_\omega > 0$,
   \begin{equation}\label{dioph2}
      \frac{C_\omega}{q_n^{2+\kappa}} 
	   < \left|\omega - \frac{p_n}{q_n}\right|
	   \leq \frac{1}{q_nq_{n+1}} < \frac{1}{q_n^2}.
   \end{equation}
   
   We will also need the following fact:

   \begin{lemma}\label{JitoLemma}(e.g. \cite{JL00})
       For an interval $I\subset\tor$, if $n$ is such that 
       $|I| > \frac{1}{q_n}$ then for any $\theta\in\tor$ there is 
       $0\leq j \leq q_n + q_{n-1} - 1$ so that $\theta + j\omega \in I$.
   \end{lemma}

   We are now ready to formulate a more detailed version of the main Theorem

   \begin{theorem}\label{main2} Assume  $f\in\mathcal{C}^\gamma(\tor)$
	  with  $1\geq \gamma >0.$ Then 
     \begin{enumerate}
      \item If $\omega$ is $\kappa$-Diophantine, $\kappa >0,$ and $\gamma>\frac{1}{2},$ then
        \begin{equation}
	     S\left(\frac{p_n}{q_n}\right)\cap L_+(\omega)
                 \to S\left(\omega\right) \cap L_+(\omega). 
        \end{equation}
        for $p_n/q_n$ the sequence of continued fraction approximants of $\omega.$
      \item If $\omega$ is not
        $\kappa$-Diophantine  with $\kappa = \gamma^{-1}-1, then $
	\begin{equation}\label{irrmc1}
          S\left(\frac{p_n}{q_n}\right)\to S\left(\omega\right) 
        \end{equation}
     \end{enumerate}
     for a subsequence of approximants
   \end{theorem}

   \begin{remark}
     \begin{enumerate}
      \item  Thus, for Lipshitz $f$ (\ref{irrmc1}) holds for a.e. $\omega$ 
             (all except possibly for the bounded type). 
             This is already implicit in \cite{L1993}.
      \item Theorem~\ref{main2} certainly implies Theorem~\ref{maint}
     \end{enumerate}
         
   \end{remark}

\section{Uniform upper semicontinuity of the upper Lyapunov exponent}
\label{usule}
   This section is devoted to some fundamental properties of the Lyapunov
   exponent in the general setting. 
   It is well known that the Lyapunov exponent of
   ergodic cocycles is upper semicontinuous.
   For a uniquely ergodic underlying dynamics,
   Furman \cite{furman97} has shown, by a subadditivity
   argument originally used by
   Katznelson and Weiss \cite{KW} to prove
   Kingman's ergodic theorem, that rate of convergence of a
   cocycle from above can be bounded uniformly in the phase.
   Now we investigate the coincidence of these properties.

   Assume $(X,T,\mu)$ is an ergodic Borel probability space.
   We use the notation $\{f\}$
   for a sequence $\left(f_n\right)\in \mathcal{C}(X,\real)\cap L^1(X,\mu)$
   which is a continuous subadditive
   cocycle with respect to $T,$ that is
   \[ f_{n+m}(x) \leq f_n(x) + f_m(T^nx). \]
   The category of continuous subadditive cocycles
   will be denoted $\Gamma(X)$.
   We define the Lyapunov exponent as,
   \[ \Lambda(f) = \lim_n \frac{1}{n}\int_X f_n(x) \mu(\de x). \]
   By Kingman's subadditive ergodic theorem we have,
   for almost all $x\in X$,
   \begin{equation}\label{lyapunov2}
 \lim_n \frac{1}{n} f_n(x) = 
      \inf_n \frac{1}{n}\int_X f_n(x) \mu(\de x) = \Lambda(f).
   \end{equation}
   To proceed it will be useful to introduce a metric on $\Gamma(X)$.
   For two continuous cocycles $\{g\},\{f\}\in \Gamma(X)$ define
   \be \label{cocydist}
           \de\left(\{g\},\{f\}\right) = \sum_{n\geq 1}
           \frac{1}{2^n}\frac{\|g_n - f_n\|_0}{1+ \|g_n-f_n\|_0}, \ee
where the norm $\|f\|_0 = \max_\theta |f(\theta)|.$ 
   Then $\left(\Gamma(X),\de\right)$ is a metric space.
   Since for any $n$ the map $\{f\}\to \frac{1}{n}\int_X f_n$
   is continuous in $\left(\Gamma(X),d\right)$,
   it follows that the infimum 
   \[\{f\}\to\inf_n \frac{1}{n}\int_X f_n \mu(\de x) = \Lambda(f)\]
   is upper semicontinuous in
   $\left(\Gamma(X),\de\right)$.\footnote{This is true for general
   $L^1$ cocycles, with no continuity required.}
   On the other hand, for a fixed cocycle over uniquely ergodic
   dynamics the convergence
   is uniform in the phase.
   \begin{theorem}{\bf (Furman \cite{furman97})} \label{furmanT97}
      Let $\{f\}$ be a continuous subadditive cocycle
      on a compact uniquely ergodic space.
      Given $\epsilon>0$, there exists $n_\epsilon$
      so that for $n>n_\epsilon$ for any $x\in X$ we have
      \[\frac{1}{n}f_n(x) < \Lambda(f) + \epsilon .\]
   \end{theorem}

   In the following theorem we combine these properties to obtain uniform
   uppersemicontinuity in both the cocycle and phase. Note that our
   simple proof is  self-contained, and except for a basic result in
   \cite{CS} (requiring unique ergodicity) it only uses
   compactness,continuity and subadditivity. In particular, it gives a
   significantly streamlined proof of Theorem \ref{furmanT97}.\footnote{a similar
     idea has been used in more specialized settings in \cite{aj,ajs}}
   \begin{theorem}\label{uniformsemicontinuity}
      Let $(X,T,\mu)$ be  a compact uniquely ergodic
      dynamical system. 
      Then $\Lambda:\Gamma(X)\to\bbr$ is uniformly  uppersemicontinuous with
      respect to $\de,$ meaning that given $\epsilon > 0$ there exist
     $\delta_\epsilon,n_\epsilon$ such that for $g$ with
     $\de(f,g)<\delta_\epsilon$ and $n>n_\epsilon,$ for all $x\in X$,
     \[ \frac{1}{n}g_n(x) \leq \Lambda(f) +\epsilon.\]
   \end{theorem}
   \begin{proof} 

By \cite{CS} for any $x$ and $\epsilon>0$ one finds
$m(x) >0$ such that $\frac 1{m(x)}
f_{m(x)}(x)<\Lambda(f)+\epsilon$. By continuity and compactness, we
find $M<\infty$ so that for all $x$, $m(x)\leq M$ and 
$\delta_\epsilon>0$ such that for $\de(f,g)<\delta_\epsilon$,
$\frac{1}M g_{m(x)}(x)<\Lambda(f)+2\epsilon$. By subadditivity,
for $k$ large enough and any $r=0,\ldots,M$ and all $x$ one has
$$
\frac{1}{km(x)+r}g_{km(x)+r}(x)\leq \frac{1}{km(x)+r}(km(x)(\Lambda(f)+2\epsilon)+Cr)
< \Lambda(f)+3\epsilon\;.\qed
$$

\end{proof}

\section{Rate of convergence for matrix cocycles}
   \label{rocmc}
   The first application of Theorem \ref{uniformsemicontinuity}
   is to approximations of matrix cocycles.
   Consider a continuous matrix
   $A\in\mathcal{C}\left( X, {\bf GL}_{n}(\cx) \right)$
   defined on a compact uniquely ergodic space $\left(X,T,\mu\right)$.
   Let the metric on
   $\mathcal{C}\left( X, {\bf GL}_{n}(\cx) \right) $
   be defined by the norm $\|A\|_0 = \max_\theta \|A(\theta)\|$.
   Then 
   \[ \ln\|A_n(\theta)\| := \ln\| A(T^{n-1} x)\cdots A(x) \|,\; A_0=I, \]
   is a subadditive cocycle
   and its Lyapunov exponent is defined by
   \[\ly(A) = \inf_n \frac{1}{n} \int_X\ln\| A(T^{n-1} x)\cdots A(x) \|\mu(\de x).\]

   Immediately, an application of Theorem \ref{uniformsemicontinuity}
   results in uniform uppersemicontinuity of the Lyapunov exponent:
   given $\epsilon,$ for $D$ near $A$ and large $k$, we have
   \begin{equation}\label{fcocy}
     \|D_k(x)\| \leq \exp\{k(\ly(A)+\epsilon)\}
   \end{equation}
   uniformly in $x$, since $D$ in a small $C_0$ neighborhood of $A$ implies
   $ \{\ln\| D_n \|\}$ in a small $\de$ neighborhood of $ \{\ln\| A_n \|\}.$
   This observation leads to the following

   \begin{corollary}
   \label{cocyest}
      Let $\epsilon >0,$ and
      $A\in\mathcal{C}\left( X, {\bf GL}_{n}(\cx) \right).$ 
      For small enough $\delta$ and large $k_\epsilon$, if
	  \[ \|D - A\|_0  < \delta \] 
      and $k\geq k_\epsilon$, then
      \begin{equation}\label{cocyestimate}
         \| A_k - D_k \|_0 \leq 
           \delta e^{\{k(\ly(A)+\epsilon)\}}.
      \end{equation}
   \end{corollary}

   \begin{proof}
      To bound the left hand side of (\ref{cocyestimate}) we will break it into terms composed of iterates of  cocycles  $A$ and $D$. 
      We obtain this by a standard trick
      \begin{eqnarray}
        A_k(\theta) &=& A(T^{k-1}\theta)\circ A_{k-1}(\theta)\nonumber 
        =\left(A-D\right)(T^{k-1}\theta)\circ A_{k-1}(\theta) + D(T^{k-1}\theta)\circ A_{k-1}(\theta)
      \end{eqnarray}
      and iterate on the last term to retrieve the iterates of $D$
      \begin{eqnarray}
        \| A_k(\theta) - D_k(\theta) \| &\leq& 
        \sum_{0\leq \ell\leq k-1}
        \| D_\ell(T^{k-\ell}\theta)(D-A)(T^{k-1-\ell}\theta)A_{k-1-\ell}(\theta)\|
        \nonumber  \\
        & \leq & \sum_{0\leq \ell\leq k-1} \| D_\ell\|_0
         \|D-A\|_0\|A_{k-1-\ell}\|_0.\nonumber
      \end{eqnarray}
      For small enough $\delta > 0 $  we may apply Theorem
      \ref{uniformsemicontinuity} to both the $A$ and $D$ iterates in
      the last term. Particularly, setting $f(x)=\ln \| A(x)\|$ we
      obtain that
       for  $0 <\epsilon' < \epsilon$
	  there is $\|D-A\|_0<\delta$  and $k({\epsilon'})$ large so
        for $\ell \geq k({\epsilon'})$, \be \label{expl}\|D_\ell(x)\|\leq e^{k\ly(A)+\epsilon'}\ee
     
      We partition the sum accordingly, for $k > 2 k({\epsilon'})$
      \be
       \| A_k(\theta) - D_k(\theta) \| 
	      \leq \label{errorpart}
                \left(  \sum_{0\leq\ell\leq k({\epsilon'}) -1}
		        +\sum_{k({\epsilon'})\leq \ell\leq k - k({\epsilon'}) - 2} 
                +\sum_{k -k({\epsilon'})-1 \leq\ell\leq k-1 }  \right)
                 \| D_\ell\|_0\|D-A\|_0  \|A_{k-1-\ell}\|_0
    \ee
    Applying (\ref{expl}) to all iterates $A_j,D_j$ for $j>
    k({\epsilon'})$, noting that the first and last summands
    consist of $k(\epsilon')$ terms each and that for $\ell<
    k(\epsilon')$ we can bound 
     $\|D_\ell\|_0 \leq (\|A\|_0 + \epsilon)^\ell$, $\|A_{k-1 -
       \ell}\|_0 \leq (\|A\|_0 + \epsilon)^{k - 1- \ell}$ , we obtain
        \beq
         \| A_k(\theta) - D_k(\theta) \| 
          &\leq& \nonumber \delta
                 \sum_{k_{\epsilon'}\leq \ell\leq k - k_{\epsilon'}- 2}
                 \| D_\ell\|_0\|A_{k-1-\ell} \|_0 
                 + 2\delta e^{\{(k-1)(\ly+\epsilon')\}}
                 \sum_{0\leq\ell\leq k_\epsilon-1}  \left(\|A\|+\delta\right)^\ell \\
          &\leq& \nonumber \delta e^{\{(k-1)(\ly+\epsilon')\}}
                 \left( k + 2 \sum_{0\leq\ell\leq k_{\epsilon'}-1}
			     \left(\|A\|+\delta\right)^\ell \right)\\
	      &\leq&\nonumber \delta e^{\{k (\ly+\epsilon)\}}
      \end{eqnarray}
      for large enough $k > k_\epsilon$. \qed
   \end{proof}
   
   \begin{remark}
      A standard argument would easily obtain 
	  (\ref{cocyestimate}) with $\exp\{k(\ly+\epsilon)\}$ replaced by
	  $C\|A\|_0^k$. The issue here is tight control on the exponential 
	  rate of growth of the error, without assuming continuity of $\ly$.
   \end{remark}

\section{H\"older Continuity in Frequency}\label{ear}

   If $I = [u,v]\subset\zee$ we write
   \[ \ham_{I;\theta} = \ham_{[u,v],\theta} := R_I\ham_\theta R_I \]
   where $R_I$ projects onto the subspace of coordinates restricted to $I$.
   The Green's function for the interval is the inverse of the 
   restriction $ G_I(i,j)  = \delta_i^T\ham_{I}^{-1}\delta_j$.
   The determinants of the truncated matrix will be labeled
   $P^E_k(\theta) := \det(\ham_{[0,k-1];\theta}-E).$
   The truncated Hamiltonian relates to the cocycle matrices by the equation
   \begin{equation}\label{cocy2det}
      A^E_k(\theta) = \left[\begin{array}{cc} P^E_k(\theta) & -P^E_{k-1}(\theta+\omega)
      \\ P^E_{k-1}(\theta) & - P^E_{k-1}(\theta+\omega)  \end{array}\right].
   \end{equation}

   The following simple lemma allows to bound $|P_k|$ from above
   uniformly in $\theta$ and for a large measure subset of the spectrum
   \begin{lemma}\label{lmuc}
      For any $\zeta,\eta > 0$ there exists a set 
      $F(\zeta,\eta)\subset S(\omega),$
      $|F(\zeta, \eta)| < \zeta,$ and 
      $ k(\omega, \zeta, \eta) = k_F$ so that
      $E\in S(\omega)\backslash F(\zeta,\eta)$ and $k>k_F$ implies
	  \begin{equation}\label{detest}
         |P^E_k(\theta)| < e^{k\left(\ly(E) +\eta\right)} .
	  \end{equation}
      Furthermore there is some $\zeta_F > 0$
      so that uniform upper convergence in the sense of 
      Corollary \ref{cocyest} holds.
      Thus, $E\in S(\omega)\backslash F(\zeta,\eta)$ implies
      if $\| D - A^E\| < \zeta_F$ and $k > k_F$
      then (\ref{cocyestimate}) holds.
   \end{lemma}
    \begin{proof}
       For all $E$ there exists $k_{E,\eta}$ and $\zeta_E$ so that
	   Corollary \ref{cocyest} holds.
	   Thus, 
	   \[|\{E: k_{E,\eta}>k\}|\to 0 \textrm{ as }k\to\infty\]
	   and 
	   \[ |\{E: \zeta_E<\delta\}|\to 0 \textrm{ as }\zeta\to 0. \]
	   Therefore, 
	   \[F(\zeta,\eta)=\{E: \zeta_E<\delta\}\cup \{E: k_{E,\eta}>k\}\]
	   for small enough $\zeta$ and large enough $k = k_F$ so 
	   that $|F(\zeta,\eta)|<\zeta$.\qed
	\end{proof}

  \subsection{The general case}
   Here we observe that a result of Avron, Mouche and Simon on
   $1/2$-H\"older continuity of the spectrum easily generalizes from
   $f\in \mathcal{C}^1$ to $\gamma$-H\"older case. 
   \begin{theorem}\label{arc}
          Suppose $f\in\mathcal{C}^\gamma(\tor)$, $1\geq\gamma > 0$.
	  Then for $E \in S(\omega)$ 
	  and for small enough $|\omega - \omega'|$,
	  there exists an $E'\in S(\omega')$ 
	  so that $|E-E'| < C_f |\omega - \omega'|^{\frac{\gamma}{1+\gamma}}$, 
	  for some constant $C_f>0$ not depending on $\omega$ or $\omega'$.
   \end{theorem}
   
   Note that by $C_f$ we mean a constant that depends only on $f.$ 
   Different such constants are denoted by the same $C_f$ in the proofs below.
   The proof is very similar to that of \cite{AMS1990}.  
   Starting with an approximate eigenfunction for $\ham_{\omega,\theta} -E$ 
   and using the same test function as in \cite{AMS1990}, upon a cutoff at
   a distance $L$ we obtain an approximate eigenfunction for 
   $\ham_{\omega',\theta'} $ with an error in the kinetic energy 
   of order $L^{-1}.$  
   The main difference is that the potential energy error is now 
   bounded by $CL|\omega-\omega'|^\gamma,$ so the choice of $L$ 
   is optimized by $L=C_f|\omega-\omega'|^{-\frac{\gamma}{1+\gamma}}.$ 

   More precisely, given $\epsilon >0$ and $E\in S(\omega),$ there exists an 
   approximate eigenfunction $\phi_\epsilon\in\ell^2(\zee)$ so that 
   $\| (\ham_{\omega,\theta} - E)\phi_\epsilon\| < \epsilon\|\phi_\epsilon\|$. 	  
   Set $g_{j,L}(n) = \left(1 - \frac{|j-n|}{L}\right)^+,$ 
   where $g^+(n)=g(n), n\geq 0$ and  $g^+(n)=0$ otherwise.

   Avron-van Mouche-Simon \cite{AMS1990} prove that for sufficiently 
   large $L$ for {\it any} bounded $f:\tor\to\bbr$  there exists $j$ 
   such that $g_{j,L}\phi_\epsilon\not= 0$,and for any $\epsilon > 0,$
   \begin{equation}\label{ams}\|(\ham_{\omega,\theta} - E)g_{j,L}\phi_\epsilon\|^2
       \leq C\left(\epsilon^2
         + L^{-2}\right)\|g_{j,L}\phi_\epsilon\|^2, 
   \end{equation}
   where $C$ is universal.
   Now let $\theta'$ be given by $\omega j +\theta = \omega' j + \theta'$.
   By the H\"{o}lder assumption on $f$ and $j-L\leq n\leq j+L$, observe that
   \[ |f(\theta + n\omega) - f(\theta' + n\omega')| 
	             \leq C_f\left(L|\omega - \omega'|\right)^\gamma\]
   Thus,
   \begin{eqnarray} 
       \|(\ham_{\omega',\theta'} - E)g_{j,L}\phi_\epsilon\|
           & \leq & \nonumber
                    \|(\ham_{\omega',\theta'} - \ham_{\omega,\theta} )g_{j,L}\phi_\epsilon\|
                   + \|(\ham_{\omega,\theta} - E)g_{j,L}\phi_\epsilon\|\\
           &\leq& \label{app}
                  \left( C_f\left(L|\omega - \omega'|\right)^\gamma + C\left(\epsilon^2
	          + L^{-2}\right)^{1/2}\right)\|g_{j,L}\phi_\epsilon\|.
    \end{eqnarray} 	    
	
    Since $\epsilon$ can be arbitrarily small, choosing 
     $L=C_f |\omega-\omega'|^{-\frac{\gamma}{1+\gamma}},$ to make both 
    addends on the right-hand side of (\ref{app}) equal,  we obtain the 
    statement of Theorem~\ref{arc} by the variational principle. \qed

\subsection{Diophantine case}

   \label{edr}
   As discussed in detail in \cite{AMS1990} (the last section), for Diophantine rotations
   $1/2$-H\"older continuity of the spectrum (the best that can be
   obtained from Theorem~\ref{arc}) is not sufficient, so  that is
   what we aim to improve.

  \begin{theorem}\label{drc}
      Suppose $\ham_{\omega,\theta}$ is an operator of the form (\ref{schrodinger1})
      where $f\in\mathcal{C}^\gamma$, $1 \geq \gamma > 0,$ $\omega\in
      [0,1]$ is $\kappa$-Diophantine, $\kappa>0.$  Fix $0<\beta <\gamma.$
      Given $\zeta>0$ there is a $B_\zeta$,
      $0 < |B_\zeta| < \zeta$ so that for 
      $E\in S(\omega)\cap L_+(\omega)\backslash B_\zeta$ and any $\omega'$ near $\omega$,
      there exists $E'\in S(\omega')$ such that
      \[ |E - E'| < C_f |\omega - \omega'|^\beta .\]
   \end{theorem}
   \begin{remark}
      The theorem holds for
      $\gamma > \beta > 0$, but the application we are interested in 
      will require $\gamma > \beta > \frac{1}{2}$.
   \end{remark}
   \begin{proof} We assume $L_+(\omega)\cap S(\omega)\neq\emptyset$
      otherwise the Theorem holds vacuously.
      Suppose $f$ is $\gamma$-H\"{o}lder. Let $0 < \beta < \gamma$.
      Let $\mathcal{E}_\chi = \{E\in S(\omega)\cap L_+(\omega): \ly(E) < \chi\}$, 
      with $\chi > 0$ so small that $|\mathcal{E}_\chi |< \frac{\zeta}{2}$.
      By upper semicontinuity of the Lyapunov exponent, the 
      Lyapunov exponent is bounded on compact sets. Let $\bar{\chi}>0$
      be an upper bound of the Lyapunov exponent on $S(\omega)$.
      Let $1> c>\frac{3}{4}.$ Choose $d$ so that $c-\frac{1}{2} > d > \frac{1}{4}$.
      Choose 
      \begin{equation}\label{partauvarsigmab}
         0 < \tau < \varsigma <
  	  \frac{\gamma - \beta}{\beta +1 -\frac{\beta}{\gamma}}\frac{d}{(1+2\kappa)};\;
           1 > b > \max (1
                  -\frac{\chi}{\bar{\chi}}\tau,c)\textrm{ and } b<a<1.
       \end{equation}
       Finally, let $\eta>0$ be such that 
       \begin{equation}\label{pareta}
          0 < \eta < \min \{\chi\tau - \bar{\chi}(1- b),
             {\chi}{(1-a)}, \chi (c-d-\frac{1}{2}) \}.
       \end{equation}
       Define
       $B_\zeta = \mathcal{E}_\chi \cup F(\zeta/2 , \eta)$
       with $F(\cdot,\cdot)$
       from Lemma \ref{lmuc} with associated $k_F$ and $\delta_F$.
       Take $E\in S(\omega)\cap L_+(\omega)\backslash B_\zeta$.
       We now find an $N$th degree trigonometric polynomial $f_N$ that approximates $f.$ 
       Namely, for $\gamma$-H\"{o}lder functions $f$, we have
       \[\|f_N - f\| < C_f N^{-\gamma}        \]
       where 
       \[  f_N(\theta) :=
           K_N*f(\theta) =
           \sum_{-N\leq j\leq N} 
           \left(1 - \frac{|j|}{N+1}\right)\hat{f}(j)e^{ij\theta}, \]
       $K_N$ being Fejer's summability kernel, 
%
       see for example \cite{katznelson}.

       Set
       \[ N = \exp\left\{ \chi k\frac{\tau}{\gamma} \right\}\]
       and let $A^{(N),E}$ be the cocycle matrix defined by 
       the potential determined by the sampling function $f_N$.
 
       For a 
       map $B:\tor\to SL_2(\real)$ and associated cocycle,
        \begin{equation}\label{vsets}
	    V_k\left(t,B\right) = 
		\left\{\theta\in\tor:\frac{1}{k}\ln\|B_k(\theta)\| > t\right\}\subset\tor.
        \end{equation}
        The measure of this set, for large $k,$  can be bounded 
        below by use of Corollary \ref{cocyest}. 
        Indeed, for $k > k_F$ we have for all $\theta$,
        $\frac{1}{k}\ln\|A^E_k(\theta)\| < \ly(E)+\eta$,
        thus using (\ref{avginf}) and (\ref{fcocy}),
        \begin{equation}\label{lebound1}
	     \ly(E) \leq \int_\tor \frac{1}{k}\ln\|A^E_k(\theta)\|\de\theta 
	     \leq |V_k\left(a\ly(E),A^E\right)|(\ly(E)+\eta) 
		       + \left(1 - \left|V_k\left(a\ly(E),A^E\right)\right|\right)a\ly(E),
        \end{equation}
        the lower bound on the measure of $V_k$ follows immediately, 
        \be\label{lebound}
         \frac{(1 - a)\ly(E)}{(1-a)\ly(E) + \eta} \leq \left|V_k\left(a\ly(E),A^E\right)\right|.
        \ee
        Furthermore, we make the following claim regarding the sets $V_k(\cdot,\cdot)$ 
        for $k > k_G = \max\left\{k_F,k_{a,b,c}\right\}$, and 
        $|E - \bar{E}| ,|E - \bar{\bar{E}}| < \exp\{- \chi \tau k\}$,
        \begin{equation}\label{incls}
           V_k(a\ly(E),A^E) \subset V_k(b\ly(E),A^{(N),\bar{E}}) \subset
           V_k(c\ly(E),A^{\bar{\bar{E}}}).
        \end{equation}
        The left inclusion of (\ref{incls}) follows from the approximation,
        \begin{eqnarray}\nonumber 
	     \theta&\in& V_k \left(a\ly(E),A^E\right)
		      \implies \\ 
              \left\|A^{(N),\bar{E}}_k(\theta)\right\|
	          &\geq&\nonumber \left\|A^{E}_k(\theta)\right\| -
                \left\|A^{E}_k(\theta) - A^{(N),\bar{E}}_k(\theta)\right\|  \\
          &>&\nonumber
               e^{ak\ly(E)} - \left(\left|E - \bar{E}\right|+ C_f N^{-\gamma}\right)e^{k(\ly(E)+\eta)}
		  > e^{ak\ly(E)} - Ce^{k(\ly(E)+\eta - \chi\tau)} > e^{bk\ly(E)}.
        \end{eqnarray}
        The second  inequality follows from the definition of $\ref{vsets}$ and an
        application of Corollary \ref{cocyest},  
        the next inequality is immediate from the choice of $\bar{E}$ and $N$,
        finally the by the choice of parameters in (\ref{pareta})
        we have $\ly(E) + \eta - \tau \chi < b\ly(E)$ so that the final inequality holds.
        The right inclusion of (\ref{incls}) is similar, with comparisons
        (applications of Corollary \ref{cocyest}) made to $A^E$,
	\begin{eqnarray}\nonumber
	     \theta&\in& V_k \left(b\ly(E),A^{(N),\bar{E}}\right)\implies \\
             \left\|A^{\bar{\bar{E}}}_k(\theta)\right\|
             &\geq&\nonumber
                \left\|A^{(N),\bar{E}}_k(\theta)\right\| 
                      - \left\|A^{(N),\bar{E}}_k(\theta) - A^{E}_k(\theta)\right\|
                      - \left\|A^{E}_k(\theta) -  A^{\bar{\bar{E}}}_k(\theta)\right\| \\
             &>&\nonumber e^{bk\ly(E)} - (C_fN^{-\gamma} +
              \left|\bar{E} - E\right| + \left|E -\bar{\bar{E}}\right|)e^{k(\ly(E)+\eta)}\\
             &>&\nonumber e^{bk\ly(E)} - Ce^{k(\ly(E)+\eta - \chi\tau)} > e^{ck\ly(E)},
        \end{eqnarray}
	  again using (\ref{pareta}) to obtain the final inequality.
          Using the inclusion (\ref{incls}) and the lower bound on measure (\ref{lebound}) we have
          \begin{equation}\label{lbound}
            \left|V_k(b\ly(E),A^{(N),\bar{E}})\right| 
              \geq \frac{\chi}{\chi + \eta/(1-a)} \geq \frac{1}{2},
          \end{equation}
          with the final inequality following from (\ref{pareta}).
	  Thus $ V_{k}(b\ly(E),A^{(N),E})$,
	  being defined by a polynomial of order $4k\exp\{\chi k\tau/\gamma\}$,
	  contains an interval of length $\exp\{-\chi k\varsigma/\gamma\}$,
	  for sufficiently large $k$.
          It follows from (\ref{incls}) that $V_k(c\ly(E),A^{\bar{E}}) $
	  also contains an interval $I$ of length $\exp\{- \chi k\varsigma/\gamma\}$.

          Now we move on to constructing the approximate eigenfunction.
          Let $E_0$ be a generalized eigenvalue of $\ham_{\omega,\theta}$
	  so that $|E-E_0| < e^{-(\bar{\chi} +\eta)k}$,
	  with generalized eigenvector $\psi$.
          For spectrally a.e. $E,$
	  $|\psi(x)| = o( (1 + |x|)^{1/2+\epsilon})$ (known as
          Schnol's Theorem, see for example \cite{K07}), 
          so we assume $E_0$ is such a value. 
          Thus there exists an $x_m$ so that 
          \[ \frac{|\psi(x_m)|}{|x_m| + 1} = \max_x\left( \frac{|\psi(x)|}{|x| + 1} \right)
                \geq  \frac{|\psi(x)|}{|x| + 1} \]
          for all $x\in\zee$.
	  Let $\psi$ be normalized so that,
	    \[ \frac{|\psi(x_m)|}{|x_m| + 1} = 1 .\]

          The sublinear growth property together with the convergence properties of cocycles we have
          discussed forces $\psi$ to take on small values at
          controlled distances, allowing as to make an effective
          cutoff, as we will now show.
          Using the Diophantine property (\ref{dioph2})
          for $\omega$, we find a denominator of an approximant $q_n$ such that
	  \begin{equation}\label{diophqn}
	      |I|^{-1}  <   \exp\{ k\chi\varsigma/\gamma\} 
		   \leq q_n < \exp\left\{k\chi\varsigma(1+2\kappa)/\gamma\right\}.
	  \end{equation}
where $I \subset V(c\ly(E),A^E)$ is an interval discussed in the reasoning after (\ref{lbound}).
          Using Lemma \ref{JitoLemma}, applied to the interval  $I$ 
	  there exists an $x'_1$, with $x_m - 2q_n - k\leq x'_1 < x_m - k$,
          so that $T^{x'_1}\theta \in I \subset V(c\ly(E),A^E)$. 
	  Similarly, there exists $x'_3$, with $x_m  < x'_3 \leq x_m  + 2q_n$,
          such that $T^{x'_3}\theta \in I$. 
          The need for an upper bound on $q_n$ will arise later.
          The lower bound on the norm of $A^E$ at $T^{x'_1}\theta $
          (that follows from (\ref{vsets}) implies by the form of the cocycles of
          $A^E$ in (\ref{cocy2det}) that for
	  $x'_1 = x_1\textrm{ or } x_1 = x_1' - 1$ and $k_\ell= k, k-1,\textrm{ or } k-2$ ,
	  we have
	  \[ |P_{k_\ell}^{E_0}(T^{x_1}\theta)|  > \frac{1}{4}e^{c\ly(E)k}. \]
          Similarly for $x_3 = x_3'$ or $x_3 = x_3'-1$ and $ k_r = k, k-1,\textrm{ or } k-2$ 
           \[ |P_{k_r}^{E_0}(T^{x_3}\theta)| > \frac{1}{4}e^{c\ly(E)k}. \]
	  Let \[x_\ell = x_1 + \left[ \frac{k_\ell}{2}\right];\hspace{.1in}
	        x_r = x_3 + \left[ \frac{k_r}{2}\right]. \]
	  Set also $x_2 = x_1 + k_\ell-1$ and $x_4 = x_3 + k_r-1$.
	  Using Cramer's rule, as in \cite{J99} 
          \begin{equation}\label{cramer2}
	    | G_{[x_1,x_2]}^{E_0}(x_\ell,x_1) | 
	    < \frac{|P^{E_0}_{(x_{2}-x_\ell)}(T^{x_\ell+1}\theta)|}
                         {|P^{E_0}_{k_\ell}(T^{x_1}\theta)|}
	    < C\frac{\left(1 + \exp\left\{-(\bar{\chi}+\eta)k_\ell\right\}\right)
		         \exp\left\{(\ly(E)+\eta)\frac{k_\ell}{2}\right\}  }
                               {\exp\left\{ c\ly(E)k_\ell\right\}}
	           <\exp\left\{ -dk_\ell\ly(E) \right\}   
	  \end{equation}		
          similarly,
          \be\label{cramer2b}
           | G_{[x_3,x_{4}]}^{E_0}(x_r,x_3) | <  \exp\left\{ -dk_r\ly(E) \right\} 
          \ee
	  with the numerator in the second inequality bounded above with (\ref{cocyestimate}) 
	  and the last inequality following from (\ref{pareta}) for sufficiently large $k$.
	  For similar reasons  (\ref{cramer2}) also holds with
          $(x_\ell,x_1)$ replaced with $(x_\ell-1,x_1)$, $(x_\ell-1,x_2)$ 
          or $(x_\ell,x_2)$ and 
       (\ref{cramer2b})  holds with $(x_r,x_3)$  replaced by $(x_r-1,x_3)$, $(x_r-1,x_4)$ 
          or $(x_r,x_4)$. 
	  Let $\Lambda = [x_\ell,x_r]$ and let $\psi_\Lambda$ be the 
	  truncation of $\psi$ to $\Lambda$ or 
	  $\psi_\Lambda = R_\Lambda \psi$.
	  We have  $|\Lambda| \leq 4 q_n +k+2 <
          5\exp\left\{k\chi\frac{\varsigma(1+2\kappa)}{\gamma}\right\}$ 
          by the upper bound of $q_n$, (\ref{diophqn}).
	  By choice of $x_m$, and with $x_a = x_r$ or $x_\ell$,
	  \begin{eqnarray}\label{varpsi} 
	       \frac{|\psi(x_a)|}{|x_m| + 1} 
	      & = &\frac{|\psi(x_a)|}{|x_m| + 1}\frac{|x_a|+1}{|x_a|+1}
             \leq \frac{|x_a|+1}{|x_m|+1}\\ 
	     &\leq& \nonumber
                    \frac{|x_m| + |x_a- x_m|+1}{|x_m|+1}
                    \leq 1+|x_a-x_m|
           \leq 1+2q_n+k/2 < 3\exp\left\{k\chi\frac{\varsigma(1+2\kappa)}{\gamma}\right\}. 
	  \end{eqnarray}
          As a formal eigenfunction, $\psi$ satisfies, for $x_1 \leq x \leq x_2$,
	  \begin{equation}\label{greeneq}
	   \psi(x) = - G^{E_0}_{[x_1,x_2]}(x,x_1)\psi(x_1-1) 
                      - G^{E_0}_{[x_1,x_2]}(x,x_2)\psi(x_2+1), 
	  \end{equation}
	  and similarly for $x_3,x_4$.
	  Applying both (\ref{varpsi}) and (\ref{cramer2}) to (\ref{greeneq}) we obtain
          bound at an
	  end point of $\Lambda$,
	  \begin{equation}\label{epb}
	     \psi(x_\ell)
	       \leq  C(|x_m|+1)
               \exp\left\{k\chi\frac{\varsigma(1+2\kappa)}{\gamma}\right\}
               \exp\left\{-kd\ly(E)\right\}\\
	  \end{equation}
          A similar bound follows for $ \psi(x_\ell -1),$ and following the same reasoning
          on $[x_3,x_4]$ and using (\ref{cramer2}) we have similar bounds for 
          $\psi(x_r)$ and $\psi(x_r+1)$.
	  The cutoff function then satisfies,
	  \[ \left\|(\ham_{\omega,\theta} - E_0)\psi_\Lambda\right\| 
	       \leq C(|x_m| +1)
		 \exp\left\{- k\left(d\ly(E)-\chi\frac{\varsigma(1+2\kappa)}{\gamma}\right) 
              \right\}  \]
	  Define $\phi_\Lambda = \psi_\Lambda/\|\psi_\Lambda\| $.
          By the normalization of $\psi$, we have $\|\psi_\Lambda\| \geq |x_m| +1 \geq 1$ so that
	  \[ \left\|(\ham_{\omega,\theta} - E_0)\phi_\Lambda\right\|
	  \leq \frac{1}{|x_m| +1}\left\|(\ham_\omega - E_0)\psi_\Lambda\right\|
	  \leq  C\exp\left\{- k\left(d\ly(E) - \chi\frac{\varsigma(1+2\kappa)}{\gamma}\right) 
                          \right\}.  \]
	  For $\omega'\in\tor$ set
	  $\theta' = \theta -\frac{x_r + x_\ell}{2}(\omega - \omega') $.
	  Then, perturbing the Hamiltonian's frequency,
	  \begin{equation}\label{phfreq}
	  \left\|\left(\ham_{\omega,\theta} - \ham_{\omega',\theta'}\right) \phi_\Lambda \right\|
	          \leq \max_{x_\ell \leq x \leq x_r} | f(\theta' +x\omega') - f(\theta +x\omega)|
			  \leq C_f\left(|\Lambda|\cdot|\omega'-\omega|\right)^{\gamma}
			  <C_f |\omega' - \omega|^\gamma\exp\{ k\chi\varsigma(1+2\kappa) \}
	  \end{equation}
	 Thus
	  \begin{eqnarray}\label{asfpo}
	     \left\|\left(E - \ham_{\omega',\theta'}\right)\phi_\Lambda\right\|& \leq& 
		 |E-E_0|
		 +\left\|(E_0 - \ham_{\omega,\theta})\phi_\Lambda\right\|
	     +\left\|\left(\ham_{\omega,\theta} - \ham_{\omega',\theta'}\right) \phi_\Lambda\right\| \\
		 & \leq &\nonumber |E-E_0|
		   + C\exp\left\{- k\left(d\ly(E) - \chi\frac{\varsigma(1+2\kappa)}{\gamma}\right) \right\} 
		   +C_f |\omega' - \omega|^\gamma\exp\{ k\chi\varsigma (1+2\kappa) \}\\
		 & \leq & 
		     C \exp\left\{- k\left(d\ly(E) - \chi\frac{\varsigma(1+2\kappa)}{\gamma}\right) \right\} 
		   +C_f |\omega' - \omega|^\gamma\exp\{ k\chi\varsigma(1+2\kappa) \}.\label{asfpo3}
	  \end{eqnarray}
Thus, by the variation principle, there exists an $E'$ in $S(\omega')$ 
	  so that
	  \begin{equation}\label{varenergy}
	     |E' - E| \leq 
		 \left\|\left(E - \ham_{\omega',\theta'}\right)\phi_\Lambda\right\|.
	  \end{equation}
	If we take $k > k_G$  such that 
	  \[ \frac{-\beta\ln|\omega - \omega'|}
               {\chi\left(d - \frac{\varsigma}{\gamma}(1+2\kappa)\right)}
	      \leq k \leq
		  \frac{-(\gamma - \beta)\ln|\omega - \omega'|}{\chi\varsigma(1+2\kappa)},
		  \]
        which we can do, by (\ref{partauvarsigmab}), for sufficiently small
        $|\omega - \omega'|,$ we obtain $|E'-E| < |\omega -\omega'|^\beta.$ 
        The required smallness of  $|\omega - \omega'|$ depends only on chosen
         parameters, therefore on $\omega$ (through its Diophantine parameters), 
        $\beta,\zeta$ and $f.$\qed
   \end{proof}

\section{The strong continuity. Proof of Theorem~\ref{main2}}\label{com}
   This argument is very similar to that of \cite{JK2001} 
   (which in turn is a modification of the proof in \cite{L1993}). 
   First, continuity of $S(\omega)$ in Hausdorff metric \cite{AS1983} implies
   \begin{equation} \label{up}
      \limsup_{\frac{p}{q} \to \omega} S\left(\frac{p}{q}\right) 
            \subseteq \Sigma(\omega)  ~\mbox{,}
   \end{equation}
    for any irrational $\omega \in \mathbb{T}$ (inclusion holds set-wise, 
    not just a.e,  for any continuous $f$ and any sequence $\frac{p}{q} \to \omega$), which immediately implies 
    the corresponding inclusion in Theorem~\ref{main2}. 
    For the opposite inclusion we need to consider continued fraction 
    approximants $\frac{p_n}{q_n}.$ Note that because of continuity in $\theta$,
    the set $S(p_n/q_n)$ consists of at most $q_n$ disjoint intervals, say
     $S(p_n/q_n)=
       \cup_{i=1}^{{q'}_n}[a_{n,i},b_{n,i}]$, ${q'}_n\le q_n$.

    We now  treat Diophantine and non-Diophantine cases separately.

   For a Diophantine $\omega,$ Theorem~\ref{drc} implies that 
   for $n>n(\omega,\beta,\zeta,f),$  
   \[S(\omega)\cap L_+(\omega) \subset 
      \cup_{i=1}^{{q'}_n}
      [a_{n,i}-C_f|\omega-\frac{p_n}{q_n}|^\beta,
       b_{n,i}+C_f|\omega-\frac{p_n}{q_n}|^\beta] \cup B_\zeta\]
   thus
    \[|(S(\omega)\cap L_+(\omega)\backslash B_\zeta)\backslash S(p_n/q_n)| 
        < 2C_fq_n|\omega-\frac{p_n}{q_n}|^\beta\to0\]
   since $\beta>1/2.$

   Therefore, for every $\zeta>0,$ we have 
    $|S(\omega)\cap L_+(\omega)\backslash B_\zeta)\backslash
      \liminf _{p_n/q_n \to \omega}S(p_n/q_n)|=0.$
   Thus
   $$|S(\omega)\cap L_+(\omega)\backslash \cap_{\zeta>0}B_\zeta)
        \backslash \liminf _{p_n/q_n \to \omega}S(p_n/q_n)|=0,$$ 
   which gives the desired inclusion in Theorem~\ref{maint}.

      Now, consider the irrational $\omega$ so that there exists
	  a sequence of rational $\frac{p_n}{q_n}$ so that $p_n$ and $q_n$
	  are mutually prime and 
	  \begin{equation}\label{ucfr}
	     q_n^{\frac{1 + \gamma}{\gamma}}\left|\omega - \frac{p_n}{q_n}\right|\to 0,
	  \end{equation}
   so that, $\omega$ is not $\kappa$ Diophantine for
      $\kappa >\frac{1}{\gamma} -1$.
      Similar to the above calculation, we have, letting
      $S\left(\frac{p_n}{q_n}\right) = 
      \cup_{1\leq i \leq q_n'}[a_{n,i},b_{n,i}]$,
      and using Theorem \ref{arc} that
      \[ S(\omega) \subset \bigcup_{1\leq i \leq
      q_n'}\left[a_{n,i}-C_f\left|\omega-\frac{p_n}{q_n}\right|^\frac{\gamma}{1+\gamma},
      b_{n,i}+C_f\left|\omega -\frac{p_n}{q_n}\right|^\frac{\gamma}{1+\gamma}\right]   \]
      Thus, by (\ref{ucfr}),
      \[ S(\omega) \subset \liminf_{p_n/q_n \to \omega}S\left(\frac{p_n}{q_n}\right)  .\qed\]

\bibliographystyle{plain}	
\bibliography{noterefs}

\end{document}